\documentclass[aip,graphicx]{revtex4-1}

\usepackage{amsthm,amssymb, mathtools, xcolor}
\usepackage{xcolor}

\theoremstyle{definition}
\newtheorem{theo}{Theorem}[section]

\newtheorem{prop}{Proposition}[section]
\newtheorem{lemm}{Lemma}[section]
\newtheorem{ex}{Example}[section]
\newtheorem{alg}{Algorithm}[section]

\mathtoolsset{showonlyrefs=true}

\draft 

\begin{document}


\title{The ultradiscrete Toda lattice and the Smith normal form of bidiagonal matrices} 



\author{Katsuki Kobayashi}
  \email{kobayashi.katsuki.74a@st.kyoto-u.ac.jp}
\author{Satoshi Tsujimoto}%
\affiliation{ 
Department of Applied Mathematics and Physics, Graduate School of Informatics, Kyoto University, Kyoto, 606-8501, Japan
}%


\date{\today}

\begin{abstract}
  The discrete Toda lattice preserves the eigenvalues of tridiagonal matrices, and convergence of dependent variables to the eigenvalues can be proved under appropriate conditions.
  We show that the ultradiscrete Toda lattice preserves invariant factors of a certain bidiagonal matrix over a principal ideal domain,
  and prove convergence of dependent variables to invariant factors using properties of box and ball system.
   Using this fact, we present a new method for computing the Smith normal form of a given matrix.
\end{abstract}


\maketitle 

\section{Introduction}
Discrete integrable systems are closely related to the matrix eigenvalues or singular values.
For example, it is known that the discrete Lotka-Volterra system computes the singular values of bidiagonal matrices \cite{TNI}.
Iwasaki and Nakamura thus proposed the {\it mdLVs algorithm} to compute singular values of an upper bidiagonal matrix \cite{IN}.
From the viewpoint of the integrable systems, various numerical algorithms have been developed \cite{FIYIN, MN, PGR}.

The another application of discrete integrable systems is that they are related to cellular automata via a procedure called {\it ultradiscretization},
which is a limiting process to construct piecewise linear systems from discrete systems.
One of the remarkable features of ultradiscretization is that the conserved quantities and exact solutions are preserved by the limiting procedure.
The most famous example is the Takahashi-Satsuma's box and ball system (BBS) \cite{TS}, which is known to be obtained by an ultradiscretization of the KdV equation \cite{TTMS}.

As with discrete systems, attempts have been made to relate ultradiscrete systems with several matrix characteristics.
For example, Heidergott et al. \cite{HOW} introduced an ultradiscrete analogue of eigenvalues and eigenvectors, which are related to circuits weights in
a weighted digraph.
However, the relation between ultradiscrete systems and matrix characteristics over an ordinary ring (such as $\mathbb{Z}$ or $\mathbb{C}[x]$) that are not min-plus algebra is not known.
In this paper, we discuss the relationship between the ultradiscrete systems and the {\it invariant factors} of matrices over a principal ideal domain.

Any matrix over a principal ideal domain $R$ can be transformed into a particular form of diagonal matrix by unimodular transformations. That is,
for any matrix $A \in R^{m \times n}$, there exist invertible matrices $P \in R^{m \times m}$ and $Q \in R^{n \times n}$ such that
the matrix $S = PAQ$ vanishes off the main diagonal, $(e_1, e_2, ..., e_r, 0, ..., 0)$,
where $e_i$ divides $e_{i+1}$ for $1 \leq i \leq r-1$. The matrix $S$ is called the \textit{Smith normal form} of $A$ and
the quantities $e_1, e_2, ..., e_r$ are the \textit{invariant factors} of $A$.
The Smith normal form has application in many areas, including integer programming \cite{GN}, combinatorics \cite{St} and computations of homology groups \cite{DHV}.
There are many algorithms for computing the Smith normal form \cite{K, S}.
The main purpose of this paper is to show that the ultradiscrete Toda lattice can compute invariant factors of bidiagonal matrices with elements of a principal ideal domain as entries.
The key observation is that the gcd operation in principal ideal domain is equivalent to applying min operation for each irreducible factors,
which allows us to run multiple ultradiscrete Toda lattice simultaneously.
Using the main result of this paper, we present a new method for computing the Smith normal form of a given matrix.

The paper is organized as follows.
In Section 2, we review the relationship between the ultradiscrete Toda lattice and the BBS.
We also give several basic properties of the BBS that are necessary to prove the main theorem.
In Section 3,
we give the main theorem of this paper, which states that
the dependent variables of the ultradiscrete Toda lattice converge to the exponents of invariant factors of a certain bidiagonal matrix.
We also present an algorithm to compute invariant factors of bidiagonal matrices based on the main theorem.
Section 4 presents the concluding remarks.
\section{Ultradiscrete Toda lattice and the box and ball system}
Let us start with the discrete Toda lattice, which is given by
\begin{align}
  \left\{
  \begin{array}{l}
    q_n^{(t+1)} = q_n^{(t)}+e_n^{(t)}-e_{n-1}^{(t+1)} \\
    e_n^{(t+1)} = q_{n+1}^{(t)} e_n^{(t)}/q_n^{(t+1)} \\
    e_{-1}^{(t)} = e_{N-1}^{(t)} = 0
  \end{array}. \label{dtoda}
  \right.
\end{align}
We rewrite \eqref{dtoda} as
\begin{align}
  \left\{
  \begin{array}{l}
  q_n^{(t+1)} = e_{n}^{(t)} + \cfrac{\prod_{j=0}^{n} q_j^{(t)}}{\prod_{j=0}^{n-1}q_j^{(t+1)}} \\
  e_n^{(t+1)} = e_n^{(t)}q_{n+1}^{(t)}/q_{n}^{(t+1)} \\
  e_{-1}^{(t)} = e_{N-1}^{(t)} = 0
\end{array}, \label{eq2}
\right.
\end{align}
by which we can compute the discrete Toda lattice without subtractions.
Let us recall the ultradiscretization.
Suppose we have positive variables $a_1, a_2, a_3, a_4$ and $a_5$ and the relation
\begin{align}
  a_1 = \cfrac{a_2(a_3+a_4)}{a_5}. \label{u1}
\end{align}
We substitute $a_i = e^{-A_i/\epsilon}$ into the relation \eqref{u1},
and take the limit as $\epsilon \to +0$. Then we obtain
\begin{align}
  A_1 = A_2 + \min(A_3, A_4) - A_5. \label{u2}
\end{align}
We see that the transition from equation \eqref{u1} to \eqref{u2} is the same as replacing
$(\times, /, +)$ with $(+, -, \min)$, respectively.
We ultradiscretize \eqref{eq2} to yield
\begin{align}
  \left\{
  \begin{array}{l}
  Q_n^{(t+1)} = \min\left(E_{n}^{(t)},~\sum_{j=0}^{n} Q_j^{(t)}-\sum_{j=0}^{n-1}Q_j^{(t+1)}\right) \\
  E_{n}^{(t+1)} = E_{n}^{(t)} + Q_{n+1}^{(t)}-Q_{n}^{(t+1)}  \\
  E_{-1}^{(t)} = E_{N-1}^{(t)} = +\infty
  \end{array}. \label{udtoda}
  \right.
\end{align}
System \eqref{udtoda} is known as the {\it ultradiscrete Toda lattice} (ud-Toda lattice), which can be considered to be a time evolution of the \textit{box and ball system} (BBS).
The following property is easy to prove, but important in our study.
\begin{prop}
  The ud-Toda lattice \eqref{udtoda} defines the map
\begin{align}
  \begin{array}{ccc}
  \hspace{15pt}(\mathbf{R_{\geq 0}})^{2N-1} & ~\longrightarrow & \hspace{15pt}(\mathbf{R_{\geq 0}})^{2N-1} \\
  \rotatebox{90}{$\in$} & & \rotatebox{90}{$\in$} \\
    (Q_0^{(t)}, ..., Q_{N-1}^{(t)}, E_0^{(t)}, ..., E_{N-2}^{(t)})  &~ \longmapsto & (Q_0^{(t+1)}, ..., Q_{N-1}^{(t+1)}, E_0^{(t+1)}, ..., E_{N-2}^{(t+1)})
  \end{array}
\end{align}
where $\mathbf{R}_{\geq 0}$ denotes the set of nonnegative real numbers.
\end{prop}
\begin{proof}
  Clearly, $Q_0^{(t+1)} = \min(E_0^{(t)}, Q_0^{(t)}) \geq 0$. Suppose we have proved that $Q_0^{(t+1)}, Q_1^{(t+1)}, ..., Q_{n}^{(t+1)} \geq 0$ and $E_0^{(t+1)}, Q_1^{(t+1)}, ..., E_{n-1}^{(t+1)} \geq 0$ for some $n \geq 1$.
  We will show that $E_n^{(t+1)}, Q_{n+1}^{(t+1)} \geq 0$.
  From \eqref{udtoda}, it follows that
  \begin{align}
    E_n^{(t+1)} &= E_n^{(t)} + Q_{n+1}^{(t)} - Q_n^{(t+1)} \\
                &= E_n^{(t)} - \min\left(E_{n}^{(t)},~\sum_{j=0}^{n} Q_j^{(t)}-\sum_{j=0}^{n-1}Q_j^{(t+1)}\right) + Q_{n+1}^{(t)} \geq 0.
  \end{align}
  Similarly, $E_n^{(t)} \geq 0$ and
  \begin{align}
    \sum_{j=0}^{n+1} Q_j^{(t)} - &\sum_{j=0}^{n}Q_j^{(t+1)} = Q_{n+1}^{(t)} + \sum_{j=0}^{n} Q_j^{(t)} - \sum_{j=0}^{n-1} Q_{j}^{(t+1)} - Q_n^{(t+1)} \\
                                                           &= Q_{n+1}^{(t)} + \sum_{j=0}^{n} Q_j^{(t)} - \sum_{j=0}^{n-1} Q_{j}^{(t+1)} -  \min\left(E_{n}^{(t)},~\sum_{j=0}^{n} Q_j^{(t)}-\sum_{j=0}^{n-1}Q_j^{(t+1)}\right) \geq 0.
  \end{align}
  Hence, $Q_{n+1}^{(t+1)} \geq 0$.
\end{proof}
In what follows, we assume $Q_n^{(0)}, E_n^{(0)} \in \mathbb{Z}_{\geq 0}$ for all $n$. In this case,
we have $Q_n^{(t)}, E_n^{(t)} \in \mathbb{Z}_{\geq 0}$ for all $t$ by Proposition 2.1.
Before explaining the relation between the ud-Toda lattice and the BBS, let us review BBS as a dynamical system on $01$-sequences.
 Here, $u = (u_n)_{n \in \mathbf{Z}} \in \{0, 1\}^{\mathbf{Z}}$ denote an infinite $01$-sequence.
We also suppose $u_n = 0$ for all but finitely many $n \in \mathbf{Z}$. We call `$0$' and `$1$' an empty box and a ball, respectively.
We define the operator $T \colon \{0, 1\}^{\mathbf{Z}} \to \{0, 1\}^{\mathbf{Z}}$ by
\begin{align}
  (Tu)_n = \min \left\{ 1 - u_n, \sum_{m= -\infty}^{n-1} (u_m - (Tu)_m) \right\}. \label{udkdv}
\end{align}
Equation \eqref{udkdv} is called {\it ultradiscrete KdV equation}.
We regard the successive application of the operator $T$ as the time evolution of the BBS.
An example of time evolutions of the BBS is as follows:
  \begin{align}
    u:~\cdots011110001110010000000000000000\cdots \\
    T^1u:~\cdots000001110001101110000000000000\cdots \\
    T^2u:~\cdots000000001110010001111000000000\cdots \\
    T^3u:~\cdots000000000001101100000111100000\cdots \\
    T^4u:~\cdots000000000000010011100000011110\cdots
  \end{align}
If there are sufficiently many empty boxes before and after a block of consecutive balls,
then the block of balls propagates to the right at the same speed as the block length.
We call such a block of consecutive balls a \textit{soliton}.
As seen in the example above 
the larger soliton overtakes the smaller soliton, and the soliton amplitudes do not changed after collisions.
The length of solitons remains unchanged over time (see proof in \cite{TNS}). The conservation of the length of solitons plays a key role in our study.

Let us consider the BBS with $N$ solitons. System \eqref{udtoda} is regarded as the time evolution of the BBS by the identification \cite{NTT},
\begin{itemize}
\item $Q_n^{(t)}\colon$the length of the $(n+1)$-st block of consecutive balls at time $t$ and
\item $E_n^{(t)}\colon$the number of empty boxes between the $(n+1)$-st and the $(n+2)$-nd blocks of consecutive balls at time $t$.
\end{itemize}
For example, in the example above,
we set initial values as follows:
\begin{align}
  Q_0^{(0)} = 4, \quad Q_1^{(0)} = 3, \quad Q_2^{(0)} = 1, \quad E_0^{(0)} = 3, \quad E_1^{(0)} = 2.
\end{align}
The conserved quantities of the BBS can be expressed in terms of the dependent variables of the ud-Toda lattice as follows.
First, we define
\begin{align}
  &W_{2i+1}^{(t)} = Q_i^{(t)}, \quad i = 0, ..., N-1, \\
  &W_{2i}^{(t)} = E_i^{(t)},\quad i = 0, ..., N-2,
\end{align}
and
\begin{align}
  uC_1 &= \min_{1 \leq j_1 \leq 2N-1} W_{j_1}^{(t)}, \\
  uC_2 &= \min_{1 \leq j_1 < j_2 - 1 \leq 2N -1} (W_{j_1}^{(t)} + W_{j_2}^{(t)}), \\
  &\vdots \\
  uC_{l-1} &= \min_{1 \leq j_1 < j_2 - 1 < \cdots < j_l - l + 1  \leq 2N -1} (W_{j_1}^{(t)} + W_{j_2}^{(t)} + \cdots + W_{j_l}^{(t)}), \\
  &\vdots \\
  uC_N &= \min_{ 1 \leq j_1 < j_2 - 1 < \cdots < j_N - N + 1  \leq 2N -1} (W_{j_1}^{(t)} + W_{j_2}^{(t)} + \cdots + W_{j_N}^{(t)}).
\end{align}
Then the following proposition holds.
\begin{prop}[Ref. \onlinecite{TNS}]
  The $N$ independent conserved quantities for the ud-Toda lattice \eqref{udtoda} are given by $uC_1, uC_2, ..., uC_N$.
\end{prop}
The dependent variables $Q_0^{(t)}, \cdots Q_{N-1}^{(t)}, E_0^{(t)}, ..., E_{N-2}^{(t)}$ of the ud-Toda lattice satisfy the following lemma.
\begin{lemm}[Ref. \onlinecite{TNS}]
  There exists a positive integer $T$ such that for all $t > T$,
  \begin{align}
    Q_0^{(t)} \leq Q_1^{(t)} \leq \cdots \leq Q_{N-1}^{(t)}.
  \end{align}
\end{lemm}
Lemma 2.1 is often called the \textit{sorting property}.  We also need the following lemma, which follows from the Lemma 2.1.
\begin{lemm}
  There exists a positive integer $T$ such that for all $t > T$,
  the dependent variables $E_0^{(t)}, E_1^{(t)}, \cdots, E_{N-2}^{(t)}$ satisfy
  \begin{align}
    Q_i^{(t)} \leq E_i^{(t)}, \quad i = 0, 1, ..., N-2.
  \end{align}
\end{lemm}
\section{Ultradiscrete Toda lattice and invariant factors}
In this section, we give the main theorem of this paper and present an algorithm for computing the Smith normal form of a bidiagonal matrix.

Let $R$ be a principal ideal domain. Two elements $a, b \in R$ are said to be \textit{associate} if $a \mid b$ and $b \mid a$.
If $a \mid b$, then we define $b/a$ to be the element $c \in R$ such that $b = ac$.
The set of all units of $R$ is denoted by $R^{\ast}$.
Instead of using the ud-Toda lattice \eqref{udtoda} directly,
we replace $(+, -, \min)$ in \eqref{udtoda} with $(\times, /, \gcd)$, where gcd denotes the greatest common divisor.
That is:
\begin{align}
  \left\{
  \begin{array}{l}
  q_n^{(t+1)} = \gcd\left(e_{n}^{(t)},~\prod_{j=0}^{n} q_j^{(t)}/\prod_{j=0}^{n-1}q_j^{(t+1)}\right) \\
  e_{n}^{(t+1)} = e_{n}^{(t)}q_{n+1}^{(t)}/q_{n}^{(t+1)} \label{gcdtoda} \\
  e_{-1}^{(t)} = e_{N-1}^{(t)} = 0
  \end{array},
  \right.
\end{align}
where $e_n^{(t)}, q_n^{(t)} \in R$.
We call the system \eqref{gcdtoda} the \textit{gcd-Toda lattice}.
System \eqref{gcdtoda} is considered as an extended expression of the ud-Toda lattice \eqref{udtoda}.
When the dependent variables in \eqref{gcdtoda} have only one irreducible factor $p \in R$, i.e.,  $q_n^{(t)} = p^{Q_n^{(t)}}, e_n^{(t)} = p^{E_n^{(t)}}$ for a single irreducible element $p \in R$ and $Q_n^{(t)}, E_n^{(t)} \in \mathbb{Z}_{\geq 0}$,
then the exponents $Q_n^{(t)}, E_n^{(t)}$ satisfy the ud-Toda lattice \eqref{udtoda}, since $\gcd(q^a, q^b) = q^{\min(a, b)}$ for $a, b \in \mathbb{Z}_{\geq 0}$.
Thus, when the dependent variables have more than one irreducible factors, the equations \eqref{gcdtoda} is equivalent to running the ud-Toda lattice simultaneously on each irreducible factors without performing prime factorization.
This also proves that the divisions in \eqref{gcdtoda} can always be performed.
The above observation is important for connecting ultradiscrete systems and computation of invariant factors.
Let $X^{(0)} \in M(n, R)$ be a lower bidiagonal matrix, and denote elements of $X^{(0)}$ as
\begin{align}
     X^{(0)} = \left(\begin{array}{ccccc}
     q_0^{(0)}&&&& \\
     e_0^{(0)}&q_1^{(0)}&&& \\
     &\ddots&\ddots&& \\
     &&e_{N-3}^{(0)}&q_{N-2}^{(0)}& \\
     &&&e_{N-2}^{(0)}&q_{N-1}^{(0)}
   \end{array}\right). \label{mat}
  \end{align}
Suppose $q_0^{(0)}, q_1^{(0)}, \cdots, q_{N-1}^{(0)}$ and $e_0^{(0)}, e_1^{(0)}, \cdots, e_{N-2}^{(0)}$ are nonzero.
We compute $q_n^{(t)}, e_n^{(t)}$ for $t = 1, 2, ...$ by \eqref{gcdtoda}.
Then we obtain
\begin{align}
  X^{(t)} = \left(\begin{array}{ccccc}
  q_0^{(t)}&&&& \\
  e_0^{(t)}&q_1^{(t)}&&& \\
  &\ddots&\ddots&& \\
  &&e_{N-3}^{(t)}&q_{N-2}^{(t)}& \\
  &&&e_{N-2}^{(t)}&q_{N-1}^{(t)}
\end{array}\right).
\end{align}
The following theorem is the main result of this paper.
\begin{theo}
  For sufficiently large $t > 0$, the diagonal part of the matrix $X^{(t)}$ coincides with
  the Smith normal form of the initial matrix $X^{(0)}$. In other words, the dependent variables $q_0^{(t)}, q_1^{(t)}, ..., q_{N-1}^{(t)}$ of the gcd-Toda lattice \eqref{gcdtoda}
  converge to the invariant factors of $X^{(0)}$ in a finite time.
\end{theo}
Before giving a proof of Theorem 3.1, we introduce \textit{determinantal divisors}.
The $i$-th determinantal divisor $d_i(A)$ of a matrix $A$ is the gcd of all $i \times i$ minors of $A$.
The $i$-th invariant factor $s_i(A)$ of $A$ is expressed as $s_i(A) = d_i(A)/d_{i-1}(A)$, where $d_0(A) = 1$.
The determinantal divisors of the bidiagonal matrix $X^{(t)}$ are expressed in a simpler form by means of the elements of the matrix $X^{(t)}$.
\begin{lemm}
   Define the variables
   \begin{align}
     &w_{2i+1}^{(t)} = q_i^{(t)}, \quad i = 0, ..., N-1, \\
     &w_{2i}^{(t)} = e_i^{(t)}, \quad i = 0, ..., N-2.
   \end{align}
  Then, the determinantal divisors of the matrix $X^{(t)}$ are given by
  \begin{align}
    d_0^{(t)} &= \gcd_{1 \leq j_1 \leq 2N-1} w_{j_1}^{(t)}, \\
    d_1^{(t)} &= \gcd_{1 \leq j_1 < j_2 - 1 \leq 2N -1} w_{j_1}^{(t)} w_{j_2}^{(t)}, \\
    &\vdots \\
    d_{l-1}^{(t)} &= \gcd_{1 \leq j_1 < j_2 - 1 < \cdots < j_l - l + 1  \leq 2N -1} w_{j_1}^{(t)} w_{j_2}^{(t)} \cdots w_{j_l}^{(t)}, \\
    &\vdots \\
    d_{N-1}^{(t)} &= \gcd_{1 \leq j_1 < j_2 - 1 < \cdots < j_N - N + 1  \leq 2N -1} w_{j_1}^{(t)} w_{j_2}^{(t)} \cdots w_{j_N}^{(t)},
  \end{align}
  where $\gcd_{1\leq i \leq n} a_i$ denotes the gcd of all $a_1, a_2, ..., a_n$.
\end{lemm}
The above lemma can be proved easily through direct calculation.
We now return to the proof of Theorem 3.1.
\begin{proof}[Proof of Theorem 3.1]
  First, we show that the invariant factors of $X^{(t)}$ do not depend on the variable $t$.
  Let $p_1, p_2, ..., p_m$ be all irreducible elements that appear in the irreducible decomposition of
  $q_0^{(0)}, ..., q_{N-1}^{(0)}$ and $e_0^{(0)}, ..., e_{N-2}^{(0)}$, and suppose that none of them are associate to any of others. Then, no irreducible elements other than $p_1, p_2, ..., p_m$
  appear in the decomposition of $q_i^{(t)}$ and $e_i^{(t)}$ for $t \geq 1$, because system \eqref{gcdtoda} contains only multiplications, divisions, and $\gcd$ operations.
  The dependent variables $q_0^{(t)}, ..., e_{N-2}^{(t)}$ are expressed as
  \begin{align}
  &q_0^{(t)} = u_{0}^{(t)} p_0^{Q_{0,0}^{(t)}}p_1^{Q_{1,0}^{(t)}}\cdots p_m^{Q_{m,0}^{(t)}}, \quad  \cdots, \quad q_{N-1}^{(t)} = u_{N-1}^{(t)}p_0^{Q_{0,N-1}^{(t)}}p_1^{Q_{1,N-1}^{(t)}}\cdots p_m^{Q_{m,N-1}^{(t)}}, \\
  &e_{0}^{(t)} = v_{0}^{(t)}p_0^{E_{0,0}^{(t)}}p_1^{E_{1,0}^{(t)}}\cdots p_m^{E_{m,0}^{(t)}}, \quad \cdots, \quad e_{N-2}^{(k)} = v_{N-2}^{(t)}p_0^{E_{0,N-2}^{(t)}}p_1^{E_{1,N-2}^{(t)}}\cdots p_m^{E_{m,N-2}^{(t)}},
\end{align}
  where $Q_{i,j}^{(t)}, ~E_{i,j}^{(t)} \in \mathbf{Z}_{\geq 0}$, $u_{i}^{(t)}, v_{i}^{(t)} \in R^{\ast}$.
Because $q_0^{(t)}, ..., e_{N-2}^{(t)}$ satisfy the gcd-Toda lattice \eqref{gcdtoda}, exponents  $Q_{i, 0}^{(k)}, ... Q_{i, N-1}^{(k)}, E_{i, 0}^{(k)}, ..., E_{i, N-2}^{(k)}$ of a single irreducible factor $q_i$
   satisfy the ud-Toda lattice \eqref{udtoda}.
   By Proposition 2.2, we have conserved quantities of the ud-Toda lattice $uC_{i, 1}, ..., uC_{i, N}$ for each $i = 1, ..., m$.
   Therefore, we obtain conserved quantities $C_1, ..., C_N$ of the system \eqref{gcdtoda} :
   \begin{align}
    C_k = p_{0}^{uC_{0, k}}  p_{1}^{uC_{1, k}} \cdots  p_{m}^{uC_{m, k}}, \quad k = 1, 2, ..., N,
   \end{align}
   where $u_k \in R^{\ast}$.
   By Lemma 3.1, we see that $C_k$ and $d_k^{(t)}$ differ by a multiplicative factor of a unit; thus, invariant factors do not depend on the variable $t$.
   Next, we prove that the dependent variables $q_0^{(t)}, ..., q_{N-1}^{(t)}$ converge to the invariant factors.
   By Lemmas 2.1 and 2.2, there exists positive integer $T$ such that for all $t > T$, we have
   \begin{align}
     &Q_{i,0}^{t} \leq Q_{i,1}^{(t)} \leq \cdots \leq Q_{i,N-1}^{(t)}
   \end{align}
   for $i = 0, ..., m$ and
   \begin{align}
     Q_{i, j}^{(t)} \leq E_{i, j}^{(t)}
   \end{align}
   for $i = 0, ..., m$ and $j = 0, ..., N-2$. This means that
   $q_i^{(t)} \mid q_{i+1}^{(t)}$ and $q_i^{(t)} \mid e_{i}^{(t)}$ for all $i = 0, ..., N-2$ when $t > T$.
   Therefore, $X^{(t)}$ can be transformed into the Smith normal form by elementary row operations.
   This concludes the proof. \qedhere
\end{proof}
Based on Theorem 3.1, we present a new method for computing invariant factors of matrices over a principal ideal domain.
Let $A = (a_{ij})_{i,j=1}^n$ be an non-zero $n \times n$ matrix over a principal ideal domain $R$.
 We can transform $A$ by unimodular transformations to a bidiagonal matrix $B = (b_{ij})_{i,j = 1}^n$ with
$b_{11}, b_{22}, ..., b_{kk} \neq 0, b_{21}, b_{32}, ..., b_{k,k-1} \neq 0$ for some $k$ and $b_{jj} = b_{j+1,j} = 0$ for $j > k$.
The element $b_{k+1,k}$ may or may not be zero.
Suppose that the elements in the first row of matrix $A$ are not all zero.
Then we can set $a_{11} \neq 0$ by exchanging columns if necessary. If the elements in the first row of matrix $A$ are all zero,
then add a non-zero row to the first row of matrix $A$ and exchange the columns so that $a_{11} \neq 0$.

Let $d = \gcd(a_{11}, a_{12})$. Then there exist $p, q, s, t \in R$ such that
\begin{align}
  &a_{11}p + a_{12}q = d, \\
  &a_{11} = sd, \quad a_{12} = -td.
\end{align}
Let the matrix $G(1, 2)$ to be
\begin{align}
  G(1, 2) =
  \left(\begin{array}{ccccc}
  p&t&&&\\
  q&s&&&\\
  &&1&&\\
  &&&\ddots&\\
  &&&&1
  \end{array}\right).
\end{align}
Then $\det G(1, 2) = 1$ and the ($1, 2)$-entry of the matrix $AG(1, 2)$ is zero.
Continuing this procedure, we can transform $A$ by unimodular transformations to the matrix $A' = (a'_{ij})_{i, j=1}^n$ with $a'_{11} \neq 0$ and $a'_{12} = a'_{13} = \cdots = a'_{1n} = 0$.
If the elements in the second row or below of the matrix $A'$ are not all zero,
we can set $a'_{21} \neq 0$ by elementary operations on $A'$ as before with the first row of $A$ unchanged.
By applying the above procedure to the rows, we can transform $A'$ by unimodular transformations to the matrix $A''$ with $a''_{31} = a''_{41} = \cdots = a''_{n1} = 0$ and $a''_{12} = a''_{13} = \cdots = a''_{1n} = 0$.
Denote $A''$ again as $A$. If the submatrix $\widetilde{A} = (a_{ij})_{i,j = 2}^n$ of $A$ is non-zero, we apply the above transformations inductively to $A$
to obtain the bidiagonal matrix $B$ satisfying the conditions stated in the beginning.
Once the bidiagonalization is done, we can use the following algorithm by setting
\begin{align}
  X^{(0)} = \left(\begin{array}{ccccc}
  b_{11}&&&& \\
  b_{21}&b_{22}&&& \\
  &b_{32}&\ddots&& \\
  &&&b_{k-1,k-1}& \\
  &&&b_{k,k-1}&b_{kk}
\end{array}\right)
\end{align}
if $b_{k+1,k} = 0$ or
\begin{align}
  X^{(0)} = \left(\begin{array}{ccccc}
  b_{11}&&&& \\
  b_{21}&b_{22}&&& \\
  &b_{32}&\ddots&& \\
  &&&b_{k,k}& \\
  &&&b_{k+1,k}&0
\end{array}\right)
\end{align}
if $b_{k+1,k} \neq 0$. Note that the zero element at the bottom right of the latter matrix does not affect the correctness of the algorithm.
\begin{alg}
  ~
  \begin{enumerate}
    \item[(i)] For a given lower bidiagonal matrix $X^{(0)}$, set the initial values of dependent variables of \eqref{gcdtoda} as \eqref{mat}. Set $t = 0$.
    \item[(ii)] Calculate $X^{(t+1)}$ using \eqref{gcdtoda}.
    \item[(iii)] If terminating conditions $q_i^{(t+1)} \mid q_{i+1}^{(t+1)}$ and $q_i^{(t+1)} \mid e_i^{(t+1)}$ hold for all $i = 0, 1, ..., N-2$, then go to (iv),
    otherwise set $t := t+1$ and go to (ii).
    \item[(iv)] Output $q_0^{(t+1)}, q_1^{(t+1)}, ..., q_{k-1}^{(t+1)}$.
  \end{enumerate}
\end{alg}
\begin{ex}
  Let the initial matrix $X^{(0)}$ be
  \begin{align}
    X^{(0)} = \left(\begin{array}{ccc}
    2 &  &  \\
    4 & 6 &  \\
     & 3 & 9
  \end{array}\right).
  \end{align}
  Then, Algorithm 3.1 proceeds as
    \begin{eqnarray*}
      \hspace{-30pt} X^{(0)} &= \left(\begin{array}{ccc}
      2&& \\
      4&6& \\
      &3&9
    \end{array}\right),
    \quad &X^{(1)} = \left(\begin{array}{ccc}
        2&& \\
        12&3& \\
        &9&18
        \end{array}\right), \quad
      X^{(2)} = \left(\begin{array}{ccc}
      2&& \\
      18&3& \\
      &54&18
    \end{array}\right), \\
  \hspace{-30pt}   X^{(3)} &= \left(\begin{array}{ccc}
    2&& \\
    27&3& \\
    &324&18
    \end{array}\right),
    &X^{(4)} = \left(\begin{array}{ccc}
      1&& \\
      81&6& \\
      &972&18
    \end{array}\right).
  \end{eqnarray*}
   We see that the matrix $X^{(4)}$ satisfies the terminating conditions of Algorithm 3.1. Hence, the Smith normal form of the matrix $X^{(0)}$ is
   \begin{align}
     \left(\begin{array}{ccc}
      1&& \\
      &6& \\
      &&18
    \end{array}\right).
  \end{align}
\end{ex}
\section{Concluding remarks}
In this paper, we introduced the gcd-Toda lattice and showed
that its dependent variables converge to the invariant factors of a certain bidiagonal matrix over a principal ideal domain.
Based on Theorem 3.1, we presented a new method for computing invariant factors of a given matrix.
This is the first instance of the usage of ultradiscrete integrable systems in the computation of invariant factors.
However, several problems are left for future works.

First, there remains the tasks of the estimation of computational cost and comparison with existing algorithms.

Second, the extension of this algorithm for another ultradiscrete integrable system is an interesting problem.
We can introduce similar algorithms for invariant factors of a certain tridiagonal matrix, which is related to the ultradiscretization of the \textit{elementary Toda orbits} introduced by Faybusovich and Gekhtman \cite{FG}. This result will be discussed in a subsequent paper.

Third, in the case of the discrete integrable system, the conservation of eigenvalues is the direct consequence of the Lax formalism of the system.
Therefore, it is natural to ask whether there is any formulation of the ultradiscrete systems that ensures the conservation of invariant factors of a matrix.
\section*{Acknowledgements}
The research of KK was partially supported by Grant-in-Aid for JSPS  Fellows, 19J23445. The research of ST was partially supported by JSPS Grant-in-Aid for Scientific Research (B), 19H01792. This research was partially supported by the joint project ``Advanced Mathematical Science for Mobility Society'' of 
Kyoto University and Toyota Motor Corporation.

\section*{Data Availability Statement}
Data sharing is not applicable to this article as no new data were created or analyzed in this study.

%
%

%


\bibliography{paper1.bib}

\end{document}